\documentclass[aps,twocolumn,superscriptaddress]{revtex4}
\usepackage{amsfonts,amssymb,amsmath}            
\pdfoutput=1
\usepackage[]{graphics,graphicx,epsfig}            
\usepackage{amsthm}
\def\identity{\leavevmode\hbox{\small1\kern-3.8pt\normalsize1}}

\newtheorem{lemma}{Lemma}
\newcommand{\ket}[1]{\left | #1 \right\rangle}
\newcommand{\bra}[1]{\left \langle #1 \right |}
\newcommand{\half}{\mbox{$\textstyle \frac{1}{2}$}}

\newcommand{\proj}[1]{\ket{#1}\bra{#1}}
\newcommand{\Tr}{\text{Tr}}
\renewcommand{\epsilon}{\varepsilon}
\begin{document}

\title{The Role of Rotational Invariance in the Properties of Hamiltonians}
\date{\today}

\author{Alastair \surname{Kay}}
\affiliation{Centre for Quantum Computation,
             DAMTP,
             Centre for Mathematical Sciences,
             University of Cambridge,
             Wilberforce Road,
             Cambridge CB3 0WA, UK}
\affiliation{Max-Planck-Institut f\"ur Quantenoptik, Hans-Kopfermann-Str.\ 1,
D-85748 Garching, Germany}

\begin{abstract}
Is it possible to prove that the properties of Hamiltonians, such as the ground state energy, results of dynamical evolution, or thermal state expectation values, can be efficiently calculated when the Hamiltonians have physically motivated constraints such as translational or rotational invariance? We report that rotational invariance does not reduce the difficulty of finding the ground state energy of the system. Crucially, the construction it preserves the translational invariance of a Hamiltonian. The failure of the construction for the properties of thermal states at finite temperatures is discussed.
\end{abstract}

\maketitle

{\em Introduction:} In order to study the quantum systems that appear in nature, it is desirable to simulate their properties on a computer. A number of techniques, such as Density Matrix Renormalization Group and Quantum Monte Carlo simulations appear to work in a wide range of interesting cases for calculating ground state energies, dynamical evolutions and expectation values of thermal states. However, none of these schemes are guaranteed to perform efficiently. Can such guarantees be made? Recent results in the field of quantum information have allowed a classification of the difficulty of some of these problems. The dynamics of a 1D Hamiltonian can implement any quantum computation \cite{followup2}, and thus there must be instances which are hard to simulate on a classical computer (BQP$\neq$P). Finding the ground state energy of a one-dimensional Hamiltonian with nearest-neighbor interactions between 12-level systems has been shown to be QMA-complete \cite{gottesman}, which means that the result can be checked to be a solution only on a quantum computer, an even harder task (QMA$\neq$BQP). This result represents the culmination of a sequence of incremental steps including \cite{KSV02a,Kempe,oliveira-2005,Biamonte}. Under the restriction that the energy gap between the ground and excited states is constant for increasing system size, expectation values can be efficiently evaluated classically \cite{Hastings}, meaning that this restricted problem is within NP (where a solution can be verified on a classical computer).

One might ask whether other restrictions, especially those that are physically motivated, make the ground state energy, or thermal state expectation values, easier to find. Those that are of particular interest, as applied to an $N$-spin Hamiltonian $H$, are translational invariance (TI) and rotational invariance (RI),
$$
THT^\dagger=H; \quad U^{\otimes N}HU^{\dagger \otimes N}=H
$$
where $T$ is the permutation operator $T\ket{i_1i_2\ldots i_N}=\ket{i_2i_3\ldots i_Ni_1}$, each $i$ runs over the full dimension of the local Hilbert space, and the RI condition is required to hold for all single spin unitaries $U$. Perhaps such restrictions would allow us to guarantee that classical simulation is possible. For example, RI quantum states have two-body reduced density matrices which are Werner states. Given that these are a function of a single parameter, this indicates a reduction in the parameter space, thereby motivating such hopes. Nevertheless, it has been shown that the dynamics of a TI and RI Hamiltonian remain BQP-complete \cite{karl,Kay:08} and, similarly, QMA-completeness can be retained in a translationally invariant Hamiltonian if the interactions are allowed to extend over a logarithmic number of spins \cite{Kay:07} or non-locally \cite{karl}. However, for estimating the ground state energy of a local TI Hamiltonian, complexity theoretic arguments are not applicable in the same way as the non-TI case. The discussion of the complexity of a function $f(x)$ is based on how difficult it is to evaluate $f(x)$ as the size of the input increases. In particular, if the time required is upper bounded by a polynomial of $\lceil\log x\rceil=N$, then the problem is said to be in P i.e.~it requires polynomial time. Roughly speaking, in the case of ground state problems, the nature of the Hamiltonian terms replaces the function input, and the size of that input is related to the number of qubits $N$ (for example, for a nearest-neighbor Hamiltonian, one should specify a term for each neighboring pair of qubits). However, in the case of a TI system, there is no parameter that scales, except for $N$ (which is of size $\log N$, so solutions that are polynomial in $N$ are in the class EXP). Instead, for TI problems, we are interested in how long the algorithm takes to run as a function of $N$. It has recently been proven that this running time can only be polynomial if two complexity classes, BQEXP and QMA$_{\text{EXP}}$ are equal \cite{irani}. This was achieved for a nearest-neighbor Hamiltonian with a 1D geometry, although the local Hilbert space dimension of the spins is extremely large.

In this paper, we explore the influence of RI on ground state and thermal state problems, and show that for ground state properties, the case of translational and rotational invariance (TRI) is of equal difficulty to that of TI, i.e.~rotational invariance does not reduce the complexity of ground state calculations.

{\em Decoherence-Free Subsystems:} Central to any discussion of rotational invariance is the subject of decoherence-free subspaces and subsystems \cite{Zan97,Zan01b,Kni00,BRS04a}, which allow one to specify how to encode information, either quantum or classical, against globally applied noise of the form $U^{\otimes N}$, for some allowed set of unitaries $U$. An excellent review on the topic can be found in \cite{rob:revmod}. For concreteness, consider encoding information on qubits, taking $U$ to be any $SU(2)$ rotation; similar constructions can be achieved for $d$-dimensional systems subject to any $SU(n\leq d)$ rotational symmetry, or discrete subset thereof \cite{Byr06,Enk06}.

The $2^r$-dimensional Hilbert space of $r$ qubits can be decomposed into the form
$$
\left(\mathcal{H}_{\frac{1}{2}}\right)^{\otimes r}=\bigoplus_{j=0}^{r/2}{\mathcal M}_j\otimes{\mathcal N}_j,
$$
where the sizes of the subsystems are $|{\mathcal M}_j|=2j+1$ and $|{\mathcal N}_j|=c^r_j$ for even $r$, and $|{\mathcal N}_j|=c^{r-1}_{j-\frac{1}{2}}+c^{r-1}_{j+\frac{1}{2}}$ for odd $r$ (note that $j$ is an integer for even $r$, and a half-integer for odd $r$). The function $c^r_j$ is defined by
$$
c^r_j=\binom{r}{r/2-j}\frac{2j+1}{r/2+j+1}.
$$
The subsystems ${\mathcal N}_j$ are referred to as `decoherence-free' i.e.~they are not affected by collective rotations. Classical data can be safely encoded across all such subsystems. 
Since coherence is not preserved in superpositions across different subsystems, quantum information must be encoded in a single subsystem. In the limit of large $r$, the largest of these is for $j=\sqrt{r}/2$, with size $2^{r-\log_2r}$. $P^{r,j}$ will be used to denote a projector of $r$ qubits onto the ${\mathcal N}_j$ subsystem, reserving subscripts for denoting which qubits the projector is to be applied to.

For any $m>r/2+j$ qubits that are a subset of $r$ qubits
\begin{equation}
P^{r,j}P^{m,m/2}=0, \label{eqn:angmom1}
\end{equation}
This comes from the addition of angular momenta -- the block of $m$ spins has total angular momentum $m/2$, and the remaining spins cannot have more than $(r-m)/2$ angular momentum. It is impossible to get a total angular momentum smaller than the difference, which is necessarily greater than $j$ if $m>r/2+j$, and hence the two projectors have 0 overlap. By a similar argument,
\begin{equation}
P^{m,m/2}P^{2,0}=0. \label{eqn:angmom2}
\end{equation}
The decompositions for the 3- and 4-qubit cases are given in \cite{Byr06}, and one can readily verify, for example, that $P^{4,0}P^{3,3/2}=0$ or $P^{4,2}P^{2,0}=0$.

{\em Hamiltonian Properties:} Given the ability to encode information in decoherence-free subsystems, which have dimensions approaching those of entire quantum systems, it is relatively simple to convert Hamiltonians into RI Hamiltonians. Evidently, the key is to preserve the important features of the initial Hamiltonian such as which state is the ground state, the energy gap, locality of the Hamiltonian, other symmetries etc.
\begin{lemma}
A Hamiltonian $H_1=\sum_ih_i$ of $N$ spins of local Hilbert space dimension $d$ interacts with $k$-body terms which are local on a 1D lattice, where $\|h_i\|\leq 1$, $d$ and $k$ are $O(1)$, and the gap between the ground state and the first excited state is $\Delta$. There exists another Hamiltonian $H_2=\sum_ih_i'$ of $O(N)$ qubits interacting with $k'$-body local interactions and gap $\Delta$ that has the same ground state energy and is RI, with $k',|h_i'|\sim O(1)$.
\label{lemma:gs}
\end{lemma}
This problem does not, yet, incorporate an assumption of TI. With no restrictions on the energy gap $\Delta$, it is known that the ground state problem is QMA-complete for $k\geq 2$ \cite{KSV02a}, and this also holds when the lattice is restricted to a spatial geometry, such as 1D \cite{gottesman,Kay:08}. On the other hand, if $\Delta$ is a constant, then Hastings has proved that the problem is within NP \cite{Hastings}. The following proof is only restricted to the 1D case for pedagogical reasons, but the lemma applies to arbitrary geometries.
\begin{proof}
We transform any instance of a Hamiltonian $H_1$ into a RI version $H_2$ with the same energy gap $\Delta$ by finding a decoherence-free subsystem of $r$ qubits of size $d_r\geq d$. We now assume that $H_1$ was composed of spins of dimension $2^r$ rather than $d$. To ensure that the ground state and energy gap are unaffected by this mapping, we apply a local magnetic field $J\sum_{i=d}^{2^r-1}\proj{i}$, $J=2k\max_i|h_i|$ to every site to shift all other states to a higher energy above the ground state. Having done this, one can trivially map the levels 0 to $d_r-1$ of the spin to the $d_r$ levels in the decoherence-free subsystem. Hence, $H_2$ acts on $rN$ qubits with $k'=rk$-body terms which are local on the same lattice, and the energy gap is the same. Our reason for selecting $J$ in such a way is that were we to take the original ground state and remove one of the spins, the total energy cannot be more than $k\max_i|h_i|$ below the ground state energy, so replacing it with a state from one of the other subspaces must increase the energy over the ground state by at least $k\max_i|h_i|\geq\Delta$.
\end{proof}
It is important to observe that by encoding $d$-dimensional spins in the decoherence-free subsystems of qubits, expanding the number of spins we interact with ($k'>k$), we lose the ability to reassert nearest-neighbour interactions by recoding into a larger local Hilbert space dimension -- this would break the rotational invariance.

For example, the difficulty of finding the ground state energy of a RI Hamiltonian with 16-body interactions is the same as that of a non-RI Hamiltonian with 12-level systems and nearest-neighbor interactions, i.e.~QMA-complete \cite{gottesman}. Given that we are interested in classical simulation, we must ensure that the encoding process of $H_1\rightarrow H_2$ can be efficiently simulated. This is guaranteed because the encoding process is the parallel application of a quantum encoding algorithm \cite{Bac06a,Bac06b} on fixed sized blocks, and a quantum computation on a fixed sized block can be efficiently simulated on a classical computer.
\begin{lemma}
There exists a local encoding process $\mathcal{E}$, and decoding process, such that
$$
e^{-iH_1t}\ket{\psi}=\mathcal{E}^{-1}\left(e^{-iH_2t}\mathcal{E}(\ket{\psi})\right)
$$
for all times $t$ and initial states $\ket{\psi}$, i.e.~simulation of the evolution of $H_2$ is at least as hard as the simulation of the evolution due to $H_1$.
\label{lemma:dynamics}
\end{lemma}
\begin{proof}
In the proof of Lemma \ref{lemma:gs}, we constructed an encoding procedure $\mathcal{E}$ on single spins that transformed the entire Hamiltonian $H_1$ into a subspace of $H_2$. Thus, the action of $H_2$ on any state confined within this subspace will be identical. One such example is the application of the same encoding procedure to any input state $\ket{\psi}$.
\end{proof}

One might hope that the thermal state of $H_2$ at inverse temperature $\beta$ would accurately reproduce the properties of $H_1$ at the same temperature, up to encoding and decoding processes; we have designed $H_2$ so that it can be written as a direct sum of $H_1$ and other terms, all of which suffer energy penalties due to the added magnetic field. This additional energy should correspond to a suppression of their importance in the thermal state, and hence not significantly affect the values of any observables. Unfortunately, the vast number of these levels outweighs any potential saving, as we will now show. Consider an observable and its encoded counterpart
$$
\langle \hat O_1\rangle=\frac{\Tr(\hat O_1e^{-\beta H_1})}{\Tr(e^{-\beta H_1})};\quad
\langle \hat O_2\rangle=\frac{\Tr(\mathcal{E}(\hat O_1)e^{-\beta H_2})}{\Tr(e^{-\beta H_2})}.
$$
We might expect some freedom in the way we define $\mathcal{E}(O_1)$. However, if we are to avoid obscuring the original observable, we must specify that $\mathcal{E}(O_1)\ket{\psi}=0$ for any $\ket{\psi}$ which is not a state encoded from the original space. Thus, $\Tr(\mathcal{E}(\hat O_1)e^{-\beta H_2})=\Tr(\hat O_1e^{-\beta H_1})$, and hence
$$
\frac{\langle \hat O_2\rangle}{\langle \hat O_1\rangle}=\frac{\Tr(e^{-\beta H_1})}{\Tr(e^{-\beta H_2})}.
$$
This value is typically exponentially suppressed in $N$ (this is readily proven for any $J>\max_i|h_i|$), meaning that $\langle \hat O_2\rangle$ needs to be determined to exponential accuracy in order to rescale it to evaluate $\langle \hat O_1\rangle$. Consider as an example a system with local magnetic field $H_1=B\sum_iZ_i$. Each spin is encoded into 3 qubits to make it RI, so we write a local term of $H_2$ as $(BZ)\oplus(J\identity)$ where $\identity$ is the $6\times 6$ identity matrix.
$$
\frac{\langle \hat O_2\rangle}{\langle \hat O_1\rangle}=\left(\frac{2\cosh(B\beta)}{2\cosh(B\beta)+6e^{-\beta J}}\right)^N
$$
The denominator is clearly larger than the numerator due to the additional 6 levels, and the $N^{th}$ power effectively reduces the ratio to 0 for large $N$ and constant $\beta, J$. While this does not constitute a general proof that observables cannot be determined, the fact that there must always be degeneracy present when translating $H_1$ into $H_2$ (since the quantum data cannot span different subsystems) implies that this is the case.

\begin{lemma}
If $H_1$ is constrained to a 1D geometry and is TI, there also exists an instance $H_2$ which is confined to a 1D geometry, is TI and RI, and which recovers the previous equivalences of ground state energy and dynamics.
\label{lemma:trans}
\end{lemma}
In this mapping, we will transform the period 1 TI of the $d$ dimensional spins to a period 1 TI of qubits (i.e.~the Hamiltonian acting on a set of qubits is the same for all translations along the chain) {\em not} a $\lceil\log_2(d)\rceil$ periodicity.
\begin{proof}
As in Lemma \ref{lemma:gs}, find $d_r>d$ and recode $H_1$, then apply the magnetic field. In the previous case, the Hamiltonian terms $h_{1\ldots k}$ were transcribed to terms $h_{1\ldots rk}'$. Due to the TI of $H_1$, all the terms $h$ are the same, just acting on different blocks of qubits. However, now we will introduce a flag state $\ket{F}$, and a flag detection projector $P^F$, both of $F$ qubits. These RI flags will sit in front of each logical spin to herald the start of the logical spin. The state that the Hamiltonian acts on, and the projector $P^F$, will be chosen such that $P^F$ has a non-zero outcome only if it aligns exactly with the qubits of the flag state. Introducing these states alters the Hamiltonian terms, such that the transcription takes the form
\begin{eqnarray}
&&\left(\prod_{i=0}^{k-1} T^{-i(r+F)}P^F_{1\ldots F}T^{i(r+F)}\right)   \nonumber\\
&&\otimes h_{F+1\ldots F+r,2F+r+1\ldots 2(F+r),\ldots,(k-1)(F+r)+F+1\ldots k(F+r)}'.   \nonumber
\end{eqnarray}
Fig.~\ref{fig:flags}(a) depicts a suitable flag of the form
$$
P^{F}_{1\ldots F}=P^{2,0}_{1,3}P^{2,1}_{2,4}P^{2,0}_{5,6}P^{m,m/2}_{7\ldots6+m},
$$
for some even integer $m$, to be selected momentarily. It is also necessary to introduce a further magnetic field
$$
J'\sum_{i=0}^{N(F+r)-1}T^{-i}(\identity-P^F_{1\ldots F}\otimes P^{r,j}_{F+1\ldots F+r})T^i
$$
where $J'=2k'\max_i|h_i'|$ so that states that are not of the correct form are penalised. The Hamiltonian is $k'=k(F+r)$-body acting on $(F+r)N$ qubits. With periodic boundary conditions, the eigenstates of $H_2$ are $(3(m+1))^{N}(F+r)$ times more degenerate than $H_1$, corresponding to the different possible alignments of the ground state, and the different choices of state detected by $P^{2,1}$, of which there are $3$ per flag state, and the $m+1$ different choices for the $P^{m,m/2}$ term. It remains to prove that $C=T^{-n}P^{F}_{1\ldots F}T^n$, a term which will appear in the Hamiltonian, has no support on $AB$, the flag states from two neighboring encoded spins in the (ground) state of the system. These flag states constitute any of the $\text{rank}(P^F)$ orthogonal states that satisfy $P^F\ket{F}=\ket{F}$. By taking  $A=P^{F}_{1\ldots F}$ and $B=T^{-(F+r)}P^{F}_{1\ldots F}T^{F+r}$, we require $ABC=0$ for $n=1\ldots r+F-1$. Consider, first, the case of $n=1$. Here the $P^{2,0}_{1,3}$ term in C overlaps with $P^{2,1}_{2,4}$ in A, and thus $ABC=P^{2,0}_{2,4}P^{2,1}_{2,4}=0$. Similarly, the $P^{2,0}_{5,6}$ term in C hits $P^{m,m/2}$ in A for $n=2\ldots m$, and the $P^{2,0}_{2,4}$ also hits over the range $n=6\ldots m+3$, resulting in 0 overlap, Eqn.~(\ref{eqn:angmom2}). There is a symmetry between the overlaps of C and A at offset $n$ and those of B and C at offset $r+F-n$. By imposing that $m=\max(r-1,5)$, these offsets cover the entire range.

\begin{figure}[!t]
\begin{center}
\includegraphics[width=0.45\textwidth]{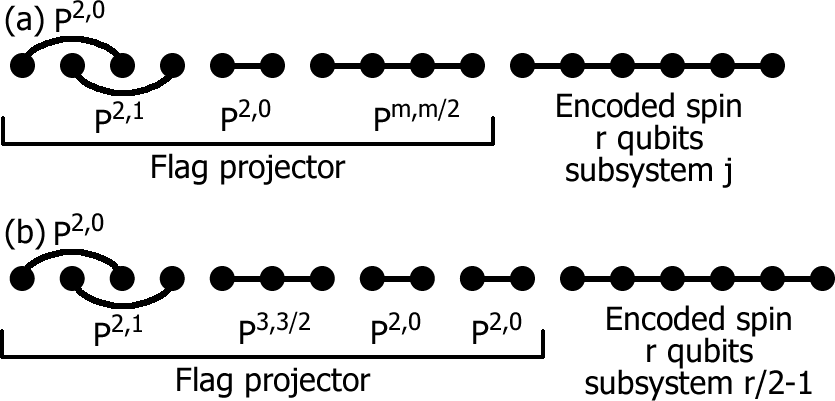}
\end{center}
\vspace{-0.5cm}
\caption{The flag projector $P^F$. There is a duality between the projectors in the Hamiltonian and the ground state of the Hamiltonian -- one selects any of the states that the flag projector projects onto. (a) A general construction that works for any values of $r$ and $j$ for sufficiently large $m$. (b) Special case construction to give $F\sim O(1)$ for $j\approx r/2$.} \label{fig:flags}
\vspace{-0.5cm}
\end{figure}
Smaller values of $m$ are often possible depending on the $\mathcal{N}_j$ subsystem that we encode into. For example, if $r/2+j+1\leq m\leq r$, then for offsets $n=m$ to $r$, the $P^{m,m/2}$ in C entirely overlaps the encoded spin of $A$, which is in the $j$ subspace, and the projection is 0, Eqn.~(\ref{eqn:angmom1}). When taken in conjunction with the symmetric version on B and C, then all offsets $n$ give 0 projection provided $r/2\geq j+6$, and we select $m=r/2+j+1$. In the limit of large $r$, with $j=\sqrt{r}/2$, the most efficient encoding into a decoherence-free subsystem, one finds that the cost of a TRI Hamiltonian scheme compared to a TI scheme is an increase in the range of the Hamiltonian terms of approximately $50\%$. For the smallest possible values of $r=3,4$, $P^F$ can be altered, replacing $P^{2,0}_{1,3}P^{2,1}_{2,4}$ with $P^{2,1}_{1,3}P^{2,0}_{2,4}$ and setting $m=r+1$, which yields a minor improvement. For large values of $j$ (of order $r/2$), one can find superior choices for the flag state, of constant size. However, there is little interest in elucidating this construction, and is relegated to the example in Fig.~\ref{fig:flags}(b).

It is immediate from the preceding arguments that there exists an encoding procedure $\mathcal{E}$ to satisfy the TI variants of Lemmas \ref{lemma:gs} and \ref{lemma:dynamics}.
\end{proof}

Upon application of this result to the Hamiltonian of \cite{followup1}, which required TI nearest-neighbor interactions of 10-level systems to implement universal quantum computation through dynamics, the result is a 38-body TRI local 1D Hamiltonian ($r=7,j=\half,m=6$). This is compared to the 106-body interactions required in \cite{Kay:08}.

{\em Conclusions:} Imposing RI on a Hamiltonian does not affect the difficulty of finding the ground state energy -- for any Hamiltonian, one can always write down a RI version that accurately reproduces all salient properties, including, in particular, TI. Hence it is not expected that TRI Hamiltonians which are local on a 1D lattice can be efficiently solved \cite{irani}. The same can be said for the $N$-representability problem as applied to a $k'$-qubit density matrix $\rho$ \cite{liu-2007-98,irani}, which should be compatible with a TRI pure state $\ket{\psi}$ of $N$ qubits. The results of \cite{aharonov} can also be applied in order to discuss the universality of Hamiltonians for Adiabatic Quantum Computation. Unfortunately, none of the results apply to the calculation of expectation values of thermal states at any finite temperature. It remains an interesting open question whether alternative formulations can prove the persistence of the complexity of calculation for finite temperatures.

{\em Acknowledgments:} N.~Schuch is thanked for a critical reading of the manuscript. This work is supported by the DFG Cluster of Excellence Munich-Centre for Advanced Photonics (MAP)
 and Clare College, Cambridge.

\end{document}